\documentclass[11pt,twoside]{article}

\usepackage[ansinew]{inputenc} 
\usepackage{amsmath,amsfonts,amssymb,amsthm}
\usepackage{pstricks,pst-node,pst-coil,pst-plot,pstricks-add}
\usepackage{graphics,geometry,epsfig}
\usepackage{bbm}

\newcommand{\field}[1]{\mathbb{#1}}
\newcommand{\N}{\field{N}}
\newcommand{\R}{\field{R}}
\newcommand{\C}{\field{C}}
\newcommand{\Z}{\field{Z}}

\newcommand{\HH}{\mathcal H}

\newcommand{\EE}{\mathcal E}
\newcommand{\FF}{\mathcal F}

\newcommand{\QQ}{\mathcal Q}

\newcommand{\eps}{\varepsilon}
\newcommand{\ph}{\varphi}

\newcommand{\sprod}[2]{\mbox{$\left\langle #1,#2 \right\rangle$}}        

\newcommand{\supp}{\operatorname{supp}}

\newcommand{\Rea}{\operatorname{Re}}

\newcommand{\curl}{\operatorname{curl}}

\newtheorem{theorem}{Theorem}[section]
\newtheorem{lemma}[theorem]{Lemma}

\newtheorem{prop}[theorem]{Proposition}
\theoremstyle{plain}

\title{The Strong-Coupling Polaron in Electromagnetic Fields}

\author{
M.~Griesemer and D.~Wellig\footnote{Supported by the German Science Foundation (DFG), grant GR 3213/1-1}\\
Universit\"at Stuttgart, Fachbereich Mathematik\\
70550 Stuttgart, Germany}

\date{}
\begin{document}
\maketitle
\begin{abstract}
This paper is concerned with Fr\"ohlich polarons subject to external electromagnetic fields in the limit 
of large electron-phonon coupling.  To leading order in the coupling constant, $\sqrt\alpha$, the ground state energy is shown 
to be correctly given by the minimum of the Pekar functional including the electromagnetic fields, provided these fields in the Fr\"ohlich model 
are scaled properly with $\alpha$. As a corollary, the binding of two polarons in strong magnetic fields is obtained.
\end{abstract}

\section{Introduction}
The purpose of this paper is to determine the ground state energy $E(A,V,\alpha)$ of Fr\"ohlich polarons subject to external electromagnetic fields
$B=\curl A$ and $E=-\nabla V$ in the limit of large electron-phonon coupling, $\alpha\to\infty$. We show that  $E(A,V,\alpha)$, to leading order in $\alpha$, is given by 
the minimum of the Pekar functional including the electromagnetic fields, provided these fields in the Fr\"ohlich model 
are scaled properly with $\alpha$. Combining this result with our
previous work on the binding of polarons in the Pekar-Tomasevich
approximation, we prove here, for the first time, the existence of
Fr\"ohlich bipolarons in the presence of strong magnetic fields.
These results were announced in \cite{GHW2012}.

The Fr\"ohlich large polaron model without external fields has only one parameter, $\alpha$, which describes the strength of the electron-phonon interaction.
Hence the ground state energy $ E(\alpha)$ is a function of $\alpha$ only, and since 
$\alpha$ is not small for many polar crystals, one is interested in
the limit $\alpha\to\infty$. It had been conjectured long ago, and
finally proved by Donsker and Varadhan \cite{DV1983}, that 
\begin{equation}\label{DV}
    E(\alpha) = \alpha^2 E_P +o(\alpha^2),\qquad (\alpha\to\infty),
\end{equation}
where $E_P$ is the minimum of the Pekar functional 
\begin{equation}\label{Pekar}
        \int|\nabla\ph(x)|^2dx - \int\!\!\int\frac{|\ph(x)|^2|\ph(y)|^2}{|x-y|}\, dxdy,
\end{equation}
constrained by 
\begin{equation}\label{norm}
     \int |\ph(x)|^2\, dx = 1.
\end{equation}
Statement \eqref{DV} has later been reproved by Lieb and Thomas who
also provided a bound on the error of the size $O(\alpha^{9/5})$ \cite{LT1997}.
An interesting application of \eqref{DV} is that it reduces the
question of bipolaron formation, in the case $\alpha\gg 1$, to the analog question regarding the
minimal energies of the Pekar and the Pekar-Tomasevich functionals. For
these effective energy-functionals the binding of two polarons follows
from a simple variational argument, provided the electron-electron
repulsion constant belongs to the lower end of its physically
admissible range. The minimizer of
\eqref{Pekar}, \eqref{norm}, which is needed for the variational
argument, is well-known to exist \cite{Lieb1977,Lions1984}. 
This line of arguments, due to Miyao and Spohn \cite{MS2007}, to our knowledge provides the only mathematically
rigorous proof of the existence of bipolarons. While it assumes that
$\alpha \gg 1$, numerical work suggest that $\alpha\geq 6.6$ may be sufficient for binding \cite{VSPD1992}.

Whether or not polarons may form bound states if they are subject to
external electromagnetic fields, e.g. constant magnetic fields, is an
interesting open question. In view of \cite{MS2007,GHW2012}, this question calls for a generalization of
\eqref{DV} to systems including a magnetic field. In the present
paper, for a large class of scalar and vector potentials 
$V$ and $A$, respectively, we establish existence of a constant $C=C(A,V)$, such that 
\begin{equation}\label{main1}
      \alpha^2 E_P(A,V) \geq E(A_\alpha,V_\alpha,\alpha) \geq \alpha^2 E_P(A,V) - C\alpha^{9/5},
\end{equation}
where $A_{\alpha}(x) = \alpha A(\alpha x)$, $V_{\alpha}(x)=\alpha^2 V(\alpha x)$ and $E_P(A,V)$ is the infimum of the generalized Pekar functional 
\begin{equation}\label{Pekar-AV}
        \int|D_A\ph(x)|^2+ V(x)|\ph(x)|^2\, dx - \int\!\!\int\frac{|\ph(x)|^2|\ph(y)|^2}{|x-y|}\, dxdy
\end{equation}
constrained by \eqref{norm}. Here $D_A = -i\nabla+A$. Non-scaled electromagnetic potentials become negligible in the 
limit $\alpha\to \infty$. In fact, we show that $\alpha^{-2}E(A,V,\alpha)\to E_P$ as $\alpha\to \infty$.

As explained above, \eqref{main1} allows us to
explore the possibility of bipolaron formation in the external fields
$A,V$. The corresponding question concerning the effective theories of
Pekar and Tomasevich with electromagnetic fields was studied in
\cite{GHW2012}. It was found, under the usual
condition on the electron-electron repulsion (see above), that two polarons will bind provided the functional
\eqref{Pekar-AV} attains its minimum, which is the case, e.g., for
constant magnetic fields and $V\equiv 0$. This
leads to our second main result, the binding of two polarons in strong
constant magnetic fields, which follows from the more general Theorem~\ref{binding-thm}, below. Of course
it would be interesting to know whether or not the binding of polarons is enhanced 
by the presence of a magnetic field, as conjectured in
\cite{BD1996}. This question is not addressed in the present paper.

The strong coupling result \eqref{DV} was generalized in
the recent work \cite{AL2012} to many-polaron systems, and one of us, 
Wellig, is presently extending this work to include magnetic
fields. In work independent and simultaneous to ours, Frank and Geisinger have
analyzed the ground state energy of the polaron for \emph{fixed} $\alpha>0$ in
the limit of large, constant magnetic field, i.e., $A=B\wedge
x/2$ and $|B|\to\infty$ \cite{FG2013}. They show that the
ground state energy, both in the Fr\"ohlich and the Pekar models, is
given by $|B| - \frac{\alpha^2}{48}(\ln |B|)^2$ up to corrections of
smaller order. The question of binding is not addressed, however, and
seems to require a similar analysis of the ground state energy of the Pekar-Tomasevich model.
For the binding of $N>2$ polarons in the Pekar-Tomasevich model with and
without external magnetic fields we refer to \cite{Lew2011} and
\cite{AG2013}, respectively. For the thermodynamic stability, the
non-binding, and the binding-unbinding transition of multipolaron
systems the reader may consult the short review  \cite{FLST2012}  and
the references therein.

\section{The Lower Bound}
\label{sec:strong-field}

In this section we study the strong coupling limit of the minimal energy of the polaron subject to given external electric and 
magnetic fields. To exhibit the general validity of the method we shall allow for fairly general electric and 
magnetic potentials $V:\R^3\to\R$ and $A:\R^3\to\R^3$. We assume that
$A_k\in L_{\rm loc}^2(\R^3)$, $V\in L_{\rm loc}^1(\R^3)$  and that for
any $\eps>0$ and all  $\ph\in C_0^{\infty}(\R^3)$,
\begin{equation}\label{V-Delta-bound}
   |\sprod{\ph}{V\ph}|\leq \eps\|\nabla\ph\|^2+C_\eps\|\ph\|^2.
\end{equation}
This is satisfied, e.g., when $V\in L^{3/2}(\R^3)+L^\infty(\R^3)$, see
\cite{RS2} or the proof of \eqref{hardy-sobolev}.
Of course, here  $\sprod{\ph}{V\ph}$ denotes a quadratic form defined by
$\sprod{\ph}{V\ph} = \int V|\ph|^2dx$. Since $\sprod{D_A\ph}{D_A\ph}$ on $H_A^1(\R^3):=\{\ph\in L^2(\R^3)\mid D_A\ph\in L^2(\R^3)\}$ is a closed quadratic form 
with form core $C_0^{\infty}(\R^3)$, it follows, by the KLMN-theorem, that $\sprod{D_A\ph}{D_A\ph}+\sprod{\ph}{V\ph}$ is the 
quadratic form of a unique self-adjoint operator $D_A^2+V$ whose form domain is  $H_A^1(\R^3)$. 
Our assumptions allow for constant magnetic fields, the case in which we are most interested.

We shall next define the Fr\"ohlich model associated with $V$ and $A$ 
through a quadratic form, which we shall prove to be semi-bounded. 
In this way the introduction of an ultraviolet cutoff is avoided. However, such a cutoff is used in the proof of semi-boundedness. 
The Hilbert space of the model in this section is the tensor product $\HH = L^2(\R^3)\otimes \FF$, 
where $\FF$ denotes the symmetric Fock space over $L^2(\R^3)$, and the form domain is $\QQ:=H_A^1(\R^3)\otimes \FF_0$ 
where 
$$\FF_0:=\big\{(\ph^{(n)})\in\FF | \ph^{(n)}\in C_0(\R^{3n}),\ \ph^{(n)}=0\ \text{for almost all}\ n\big\}.$$ 
We define a quadratic form $H$ on $\QQ$ by 
\begin{align}
   H(\psi) &:= \sprod{\psi}{(D_A^2+V)\psi} + N(\psi) + \sqrt{\alpha}W(\psi)\nonumber\\
   N(\psi)&:= \int \|a(k)\psi\|^2\,dk \label{number}\\
   W(\psi) &:= \frac{1}{\sqrt{2}\pi}\int \frac{dk}{|k|}\left(\sprod{\psi}{e^{ikx}a(k)\psi} + \sprod{e^{ikx}a(k)\psi}{\psi}\right).\label{e-p-inter}
\end{align}
Note that $a(k)$ is a well-defined, linear operator on $\FF_0$ but $a^{*}(k)$ is
not and neither is $\int |k|^{-1}e^{-ikx}a^{*}(k)\, dk$,
because $|k|^{-1}e^{-ikx}$ is not square integrable with respect to $k$.
The Theorems~\ref{limit-thm} and \ref{weak-field-thm} in the next section relate
$$
       E(A,V,\alpha) := \inf\{ H(\psi) \mid \psi\in\QQ,\ \|\psi\|=1\}
$$
to the minimum, $E_P(A,V)$, of the Pekar functional \eqref{Pekar-AV} on the unit sphere $\|\ph\|=1$.
For the proofs it is convenient to introduce a coupling constant
$\alpha$ in the Pekar functional and to
define  $E_P(A,V,\alpha)$ as the minimum of
\begin{equation*}
        \EE_{\alpha}(A,V,\ph) = \int|D_A\ph(x)|^2+ V(x)|\ph(x)|^2\, dx - \alpha\int\!\!\int\frac{|\ph(x)|^2|\ph(y)|^2}{|x-y|}\, dxdy
\end{equation*}
with the constraint $\|\ph\|=1$. We set
$\EE(\ph)=\EE_{\alpha=1}(0,0,\ph)$, which is the Pekar functional \eqref{Pekar}. It is easy to check that 
\begin{equation}\label{scaling-pekar}
      E_P(A_\alpha,V_\alpha,\alpha) = \alpha^2 E_P(A,V)
\end{equation}
where $A_\alpha(x) = \alpha A(\alpha x)$, $V_\alpha(x) = \alpha^2
V(\alpha x)$.

The number $E_P(A,V,\alpha)$ is finite because $E_P(A,V,\alpha)\geq
E_P(0,V,\alpha)$, by the diamagnetic inequality \cite{LiebLoss}, and
$E_P(0,V,\alpha)>-\infty$ by assumption on $V$ and a simple exercise
using the H\"older and Hardy inequalities. Our key result is the following lower bound on $E(A,V,\alpha)$:


\begin{prop}\label{the-lower-bound}
Suppose that $A,V$ satisfy the assumptions described above and $\beta=1-\alpha^{-1/5}$. Then 
\begin{equation}\label{fb1}
   E(A,V,\alpha) \geq \beta E_P(A,\beta^{-1}V,\alpha \beta^{-2}) -
   O(\alpha^{9/5}),\quad (\alpha\to\infty),
\end{equation}
the error bound being independent of $A$ and $V$.
\end{prop}

The proof of Proposition~\ref{the-lower-bound} is done in several
steps following \cite{LT1997}. Some of them can be taken over verbatim
upon the substitution $-i\nabla_x \to -i\nabla_x+A(x)$. Surprisingly,
the translation invariance that seemed to play some role in
\cite{LT1997} is not needed for the arguments to work. For the
convenience of the reader we at least sketch the main ideas.

To begin with, we introduce a quadratic form $\sprod{\psi}{H_\Lambda\psi}$ on $\QQ$ in terms of 
$$
    H_\Lambda:= \beta D_A^2 +V+N_{B_{\Lambda}} + \frac{\sqrt{\alpha}}{\sqrt{2}\pi}\int_{B_\Lambda}\frac{dk}{|k|}(e^{ikx}a(k) + e^{-ikx}a^{*}(k))
$$
where $\beta:=1-\frac{8\alpha}{\pi \Lambda}$, $B_\Lambda:=\{k\in\R^3:|k|\leq \Lambda\}$ and generally, for subsets $\Omega\subset\R^3$,
$$
       N_\Omega := \int_{\Omega} a^{*}(k)a(k)\, dk.
$$
The quadratic form $H_\Lambda$ is bounded below provided that $\Lambda > 8\alpha/\pi$.

\begin{lemma}\label{semibound}
In the sense of quadratic forms on $\QQ$, for any $\Lambda>0$,
$$
      H(\psi) \geq \sprod{\psi}{(H_\Lambda -\frac{1}{2})\psi}.
$$
\end{lemma}


This lemma, without electromagnetic fields, is due to Lieb and Thomas \cite{LT1997}. Its proof is based on the operator identity 
\begin{equation}\label{com-id}
     e^{ikx}a(k) = \sum_{\ell=1}^3 \Big[D_{A,\ell},\frac{k_\ell}{|k|^2}e^{ikx}a(k)\Big]
\end{equation}
where $D_{A,\ell}=-i\partial_{x_\ell} + A_\ell(x)$. Obviously, $A(x)$ plays no role in \eqref{com-id} as it drops out of the commutator,
but we need it for the estimates to follow. For any given $\Lambda>0$ and $x\in\R^3$ we define 
the Fock space operators
\begin{align*}
    \phi_{\Lambda}(x) &:= \frac{1}{\sqrt{2}\pi}\int_{B_\Lambda}(e^{ikx}a(k)+e^{-ikx}a^{*}(k))\frac{dk}{|k|}\\
    Z_\ell(x) &:= \frac{1}{\sqrt{2}\pi} \int_{\R^3\backslash B_\Lambda} \frac{k_\ell}{|k|^3}e^{ikx}a(k)\,dk
\end{align*}
and we extend them to operators $\phi_{\Lambda},Z_\ell$ on $\HH$ by setting $(\phi_{\Lambda}\psi)(x):=\phi_{\Lambda}(x)\psi(x)$, $(Z_\ell\psi)(x):=Z_\ell(x)\psi(x)$ for 
$\psi\in\HH \simeq L^2(\R^3;\FF)$.
Then, by \eqref{com-id}, the electron-phonon interaction $W$ can be
written as
\begin{equation}\label{lm-sb1}
    W(\psi) = \sprod{\psi}{\left(\phi_{\Lambda}+ \sum_{\ell=1}^3[D_{A,\ell},Z_\ell-Z_\ell^{*}]\right)\psi}.
\end{equation}
Following \cite{LT1997} one now shows that 
\begin{equation}\label{lm-sb2}
      \sqrt{\alpha}\sum_{\ell=1}^3[D_{A,\ell},Z_\ell-Z_\ell^{*}] \geq -\frac{8\alpha}{\pi \Lambda}D_A^2 - (N-N_{B_\Lambda})-\frac{1}{2}.
\end{equation}
The Lemma~\ref{semibound} follows from \eqref{lm-sb1} and \eqref{lm-sb2}. 

The next step is to \emph{localize the electron} in a box of side length $L$. To this end we define the localization function
$$
    \ph(x) = \begin{cases}\prod_{j=1}^3 \cos(\frac{\pi}{L}x_j) & \text{for}\ |x_j|\leq L/2,\\ 0 & \text{otherwise,}\end{cases}
$$
and $\ph_y(x):=\ph(x-y)$.


\begin{lemma}\label{local-electron}
For given $\Delta E>0$ define $L>0$ by $3\beta(\frac{\pi}{L})^2=\Delta E$ and let 
$\ph$ be as above. Then for every non-vanishing $\psi\in \QQ$ there exists a point $y\in \R^3$, such that $\ph_y\psi\neq 0$ and 
$$
    \sprod{\ph_y\psi}{H_\Lambda\ph_y\psi} \leq (E+\Delta E)\|\ph_y\psi\|^2,
$$
where $E:=\sprod{\psi}{H_\Lambda\psi}$.
\end{lemma}

\begin{proof}
Using $ \ph_y D_A^2 \ph_y=D_A\ph_y^2 D_A + \ph_y(-\Delta\ph_y)$ and $\beta\ph_y(-\Delta \ph_y) = \ph_y^2\Delta E$ one shows that 
$$
      \int \sprod{\ph_y\psi}{(H_\Lambda-E-\Delta E)\ph_y\psi}\, dy = 0,
$$
which proves the lemma.
\end{proof}

The Lemma~\ref{local-electron} is to be read as a bound on $E=\sprod{\psi}{H_\Lambda\psi}$ from below:
using that $H_\Lambda$ is translation invariant, except for the terms involving 
$A$ and $V$, it implies together with Lemma~\ref{semibound} that
\begin{equation}\label{local-bound}
     E(A,V,\alpha) \geq \inf_{y\in\R^3}\left(\inf_{\psi\in \QQ_L,\ \|\psi\|=1}\sprod{\psi}{H_{\Lambda,y}\psi}\right)-\Delta E-\frac{1}{2},
\end{equation}
where  $H_{\Lambda,y}$ is defined in terms of the shifted potentials 
$A_y(x) = A(x+y)$ and $V_y(x)=V(x+y)$, $\QQ_L:=(L^2(C_L)\otimes\FF)\cap\QQ$, and $C_L=\supp(\ph)\subset\R^3$ is the cube of side length $L$ centered at the origin.


The next step is the passage to \emph{block modes}. For given $P>0$ and $n\in\Z^3$ we define
\begin{align*}
    B(n) &:= \{k\in B_\Lambda \mid |k_i-n_iP|\leq P/2\},\\ 
    \Lambda_P &:= \{n\in\Z^3 \mid B(n)\neq \emptyset\}.
\end{align*}
In each set $B(n)$ we pick a point $k_n$, to be specified later, and we define
block annihilation and creation operators $a_n$ and $a_n^{*}$ by
$$
       a_n:= \frac{1}{M_n} \int_{B(n)}\frac{dk}{|k|}a(k),\qquad M_n=\left(\int_{B(n)} \frac{dk}{|k|^2}\right)^{1/2}.
$$
For given $\delta>0$ we define the block Hamiltonian
\begin{align*}
      H_{\Lambda,y}^{block} :=\ &\beta D_{A_y}^2 +V_y+ (1-\delta)\sum_{n\in\Lambda_P}a_n^*a_n\\ 
      &+ \frac{\sqrt{\alpha}}{\sqrt{2}\pi} \sum_{n\in\Lambda_P} M_n\big(e^{ik_nx}a_n + e^{-ik_nx}a^{*}_n \big),
\end{align*}
and we set $H_{\Lambda}^{block}:=H_{\Lambda,0}^{block}$. The reason for introducing block modes is well explained in
\cite{LT1997} and related to \eqref{eq:a*a}.


\begin{lemma}\label{lm:block}
In the sense of quadratic forms in $\QQ_L$, for all $(k_n)$,
$$
      H_{\Lambda,y} \geq H_{\Lambda,y}^{block} - \frac{9\alpha P^2L^2\Lambda}{2\pi\delta}.
$$
\end{lemma}

\begin{proof}
For each $n\in \Lambda_P$, by a completion of squares w.r.t. $a(k)$ and $a^*(k)$ we find, 
in the sense of quadratic forms in $\QQ_L$,
\begin{eqnarray}
\lefteqn{ \delta N_{B(n)} + \frac{\sqrt{\alpha}}{\sqrt{2}\pi} \int_{B(n)}\frac{dk}{|k|}\big(e^{ikx}a(k) + e^{-ikx}a^{*}(k)\big) } \nonumber\\
    & \geq& \frac{\sqrt{\alpha}}{\sqrt{2}\pi} \int_{B(n)}\frac{dk}{|k|}\big(e^{ik_n x}a(k) 
+ e^{-ik_n x}a^{*}(k)\big) -\frac{\alpha}{2\pi^2\delta} \int_{B(n)} \frac{dk}{|k|^2}|e^{ikx}-e^{ik_n x}|^2\nonumber\\
    & \geq&  \frac{\sqrt{\alpha}}{\sqrt{2}\pi} M_n \big(e^{ik_nx}a_n + e^{-ik_nx}a^{*}_n \big) - \frac{\alpha}{2\pi^2\delta}(\frac{3}{2}PL)^2  \int_{B(n)} \frac{dk}{|k|^2}, 
\label{c-squares}
\end{eqnarray}
where we used the definition of $a_n$ and that
$$
       |e^{ikx}-e^{ik_n x}| \leq \frac{3}{2}PL,\quad \text{for}\ x\in C_L,\ k\in B(n).
$$
After summing \eqref{c-squares} with respect to $n\in\Lambda_P$, 
the lemma follows from $\int_{B_{\Lambda}}|k|^{-2}dk=4\pi \Lambda$ and from
$a^{*}_na_n \leq N_{B(n)}$.
\end{proof}

We now use Lemma~\ref{lm:block} to bound \eqref{local-bound} from below and then we replace $\QQ_L$ by $\QQ$.
This leads to 
\begin{equation}\label{block-bound}
    E(A,V,\alpha) \geq \inf_{\psi\in \QQ, \|\psi\|=1}\sup_{k_n}\sprod{\psi}{H_{\Lambda,y}^{block}\psi} - \frac{9\alpha P^2L^2\Lambda}{2\pi\delta} - \Delta E-\frac{1}{2}.
\end{equation}
Recall that $L$ depends on $\Delta E$. It remains to compare $\sprod{\psi}{H_{\Lambda,y}^{block}\psi}$ with the minimum of the Pekar functional. This will be done 
in the proof of the following lemma using coherent states.


\begin{lemma}\label{block-lm}
Let $\mu = \alpha \beta^{-1}(1-\delta)^{-1}$. Then for every normalized $\psi\in \QQ$ and every $y\in\R^3$,
$$
       \sup_{k_n}\sprod{\psi}{H_{\Lambda,y}^{block}\psi} \geq \beta E_P(A,V,\mu) - |\Lambda_P|.
$$
\end{lemma}

\begin{proof}
Since $E_P(A_y,V_y,\mu)$ is independent of $y$ it suffices to prove the asserted inequality without the 
$y$-shift in the block Hamiltonian. Let $M=\text{span}\{|\cdot|^{-1}\chi_{B(n)} \mid n\in \Lambda_P\}$, which is a finite dimensional subspace 
of $L^2(\R^3)$. From $L^2(\R^3) = M\oplus M^{\perp}$ it follows that $\FF$ is isomorphic to $\FF(M)\otimes \FF(M^{\perp})$ with the isomorphism given by
\begin{eqnarray*}
     \Omega &\mapsto & \Omega\otimes\Omega\\
     a^*(h) &\mapsto& a^*(h_1) \otimes 1 + 1\otimes a^*(h_2)
\end{eqnarray*}
where $h_1$ and $h_2$ are the orthogonal projections of $h$ onto $M$ and $M^{\perp}$ respectively. 
Here $\Omega$ denotes the normalized vacuum in any Fock space. Note that 
\begin{equation}
    \FF(M) = \overline{\text{span}}\Big\{\prod_{n\in\Lambda_p} (a^{*}_n)^{m_n}\Omega\Big| m_n\in\N\Big\}
\end{equation}
where $ \overline{\text{span}}$ denotes the closure of the span. With respect to the factorization $\HH = \HH_M \otimes \FF(M^{\perp})$ where 
$\HH_M=L^2(\R^3)\otimes\FF(M)$, the block Hamiltonian is of the form 
$H_\Lambda^{block} \otimes 1$. To bound $H_\Lambda^{block} \otimes 1$ on $\HH_M \otimes \FF(M^{\perp})$   from below  we introduce 
coherent states $|z\rangle\in \FF(M)$ for given $z=(z_n)_{n\in\Lambda_P}$, $z_n\in\C$, by
$$
    |z\rangle:= \prod_{n\in \Lambda_P} e^{z_n a^{*}_n - \bar{z}_n a_n}\Omega.
$$
Clearly, $\sprod{z}{z}=1$ and it is easy to check that $a_n|z\rangle=z_n |z\rangle$. 
On $\FF(M)$, in the sense of weak integrals,
\begin{align}
    \int dz |z\rangle\langle z| &= 1,\nonumber\\
    \int dz (|z_n|^2-1)|z\rangle\langle z| &= a^{*}_n a_n, \label{eq:a*a}
\end{align}
where
$
        \int dz := \prod_{n\in \Lambda_P}\frac{1}{\pi}\int dx_ndy_n.
$
The second equation follows from $a^{*}_n a_n = a_n a^{*}_n-1$  and from the first one. Now suppose that $\psi\in\QQ$ and let 
$\psi_z(x) = \sprod{z}{\psi(x)}$. Then $\psi_z\in L^2(\R^3)\otimes \FF(M^{\perp})$ and 
$$
    \sprod{\psi}{H_\Lambda^{block}\psi} = \int dz \sprod{\psi_z}{(h_z\otimes 1)\psi_z}
$$
where $h_z$ denotes the Schr\"odinger operator in $L^2(\R^3)$ given by
\begin{equation*}
    h_z  = \beta D_A^2 + V+ (1-\delta)\sum_{n\in \Lambda_P} (|z_n|^2-1) \\
     +  \frac{\sqrt{\alpha}}{\sqrt{2}\pi} \sum_{n\in \Lambda_P} M_n\big(z_n e^{ik_nx} + \overline{z_n} e^{-ik_nx}\big).
\end{equation*}
Let $\widehat{\rho_z}(k):=\sprod{\psi_z}{e^{-ikx}\psi_z}$ be the Fourier transform of $\rho_z(x)=|\psi_z(x)|^2$.
By completion of the square w.r.to $z_n$ and $\overline{z_n}$ it follows that
\begin{eqnarray*}
    \sup_{k_n}\lefteqn{\int dz \sprod{\psi_z}{(h_z\otimes 1)\psi_z} }\\
    &\geq & \int dz \beta\|D_A\psi_z\|^2 +\sprod{\psi_z}{V\psi_z} - \frac{\alpha}{2\pi^2(1-\delta)} \int dz\int_{B_{\Lambda}}\frac{dk}{|k|^2} |\widehat{\rho_z}(k)|^2\frac{1}{\|\psi_z\|^2}
    -|\Lambda_P|\\
    &\geq &  \int dz\Big(\beta\|D_A\psi_z\|^2 + \sprod{\psi_z}{V\psi_z} - \frac{\alpha}{(1-\delta)\|\psi_z\|^2} \int\frac{\rho_z(x)\rho_z(y)}{|x-y|}\,dxdy\Big) - |\Lambda_P|.
\end{eqnarray*}
The integrand is readily recognized as 
$$
     \beta\|\psi_z\|^2 \EE_{\mu}(A,\beta^{-1}V,\psi_z/\|\psi_z\|),
$$ 
with coupling constant $\mu:=\alpha \beta^{-1}(1-\delta)^{-1}$. Its minimum is 
$$
        \beta\|\psi_z\|^2 E_P(A,\beta^{-1}V,\mu).
$$
Since $\int \|\psi_z\|^2dz=1$, the proof of the lemma is complete.
\end{proof}

\begin{proof}[\textbf{Proof of Proposition~\ref{the-lower-bound}}]
By \eqref{block-bound} and Lemma~\ref{block-lm} it follows that
$$
     E(A,V,\alpha) \geq \beta E_P(A,\beta^{-1}V,\mu) - |\Lambda_P| - \frac{9\alpha P^2L^2\Lambda}{2\pi\delta}-\Delta E-\frac{1}{2}
$$
where $\beta=1-\frac{8\alpha}{\pi\Lambda}$, $\mu=\alpha \beta^{-1}(1-\delta)^{-1}$ 
and $L^2=\pi^23\beta/\Delta E$. $\Lambda,\delta,P$ and $\Delta E$ are free
parameters. We choose $\Lambda=\frac{8}{\pi}\alpha^{6/5}$, $\delta=\alpha^{-1/5}$, $P=\alpha^{3/5}$ and $\Delta E=\alpha^{9/5}$. Then 
$\beta=1-\delta$ and hence the proposition follows.
\end{proof}

\section{The Strong Coupling Limit}

Equipped with Proposition~\ref{the-lower-bound} we can turn to the
proofs of the results described in the introduction in the more precise forms of 
Theorems~\ref{limit-thm} and \ref{weak-field-thm}, below.

\begin{theorem}\label{limit-thm}
Suppose the potentials $A$ and $V$ satisfy the assumptions of
Proposition~\ref{the-lower-bound}, $A_{\alpha}(x) := \alpha A(\alpha x)$ and
$V_{\alpha}(x):=\alpha^2 V(\alpha x)$. Then there exists a constant
$C=C(A,V)$ such that for $\alpha>0$ large enough,
\begin{equation*}
      \alpha^2 E_P(A,V) \geq E(A_\alpha,V_\alpha,\alpha) \geq \alpha^2 E_P(A,V) - C\alpha^{9/5}.
\end{equation*}
\end{theorem}

\begin{proof}
The first inequality follows from the well-known $E_P\geq E$, see the
proof of \eqref{inequ:pt-froe}, and from the scaling
property \eqref{scaling-pekar} of $E_P$. Using
Proposition~\ref{the-lower-bound} and \eqref{scaling-pekar}, we
see that 
\begin{equation}\label{lt1}
     E(A_{\alpha},V_\alpha,\alpha) \geq \alpha^2 \beta
     E_P(A,\beta^{-1}V,\beta^{-2}) - O(\alpha^{9/5})
\end{equation}
where $\beta=1-\alpha^{-1/5}$ and where the function $\lambda\mapsto
E_P(A,\lambda V,\lambda^2)$, as an infimum of concave functions, is concave. Therefore it has one-sided derivatives, which implies that
\begin{equation}\label{lt2}
   E_P(A,\beta^{-1}V,\beta^{-2}) \geq E_P(A,V) -O(\alpha^{-1/5}).
\end{equation}  
Combining \eqref{lt1} and \eqref{lt2} the second inequality from
Theorem~\ref{limit-thm} follows.
\end{proof}




\begin{theorem}\label{weak-field-thm}
\text{}
\begin{itemize}
\item[(a)] If $A\in L^3_{\rm loc}(\R^3)$ and $V\in
  L^{3/2}_{\rm loc}(\R^3)$ with \eqref{V-Delta-bound}, then
  $\alpha^{-2}E(A,V,\alpha)\to E_P$ as $\alpha\to\infty$.
\item[(b)] If $A=(B\wedge x)/2$ and $V\in L^{5/3}(\R^3)+L^{\infty}(\R^3)$, then 
$$
      \frac{E(A,V,\alpha)}{\alpha^2} = E_P + O(\alpha^{-1/5}),\qquad (\alpha\to \infty).
$$
\end{itemize}
\end{theorem}

The fact that non-scaled fields $A,V$ should become negligible in the
limit $\alpha\to\infty$ is seen as follows: by Proposition~\ref{the-lower-bound} and by \eqref{scaling-pekar},
$\alpha^{-2} E(A,V,\alpha)$ is bounded from above and from below by 
\begin{equation}
       E_P(A_{\alpha^{-1}}, V_{\alpha^{-1}}) \geq \frac{E(A,V,\alpha)}{\alpha^2} \geq
       \beta E_P(A_{\alpha^{-1}},\beta^{-1}
       V_{\alpha^{-1}},\beta^{-2}) -O(\alpha^{-1/5}),
\end{equation}
where $A_{\alpha^{-1}}(x) = \alpha^{-1}A(x/\alpha)$ and $V_{\alpha^{-1}}(x) = \alpha^{-2}V(x/\alpha)$.
In the limit $\alpha\to\infty$ these fields are vanishing in the sense of the 
following lemma. The theorem will thus follow from parts (b) and (c)
of Lemma~\ref{lm:weak-field2} below.
As a preparation we need:

\begin{lemma}\label{lm:weakfieldlimit}
(i) Suppose $A\in L^3_{\rm loc}(\R^3)$ and $V\in L^{3/2}_{\rm loc}(\R^3)$. Then 
\begin{align*}
A_{\alpha^{-1}} &\to 0\quad (\alpha\to \infty)\quad \text{in } L^2_{\rm loc}(\R^3),\\
V_{\alpha^{-1}} &\to 0\quad (\alpha\to \infty)\quad \text{in } L^1_{\rm loc}(\R^3).
\end{align*}
(ii) If $V=V_1+ V_2\in L^{5/3}(\R^3)+L^{\infty}(\R^3)$, then for all $\ph\in H^1(\R^3)$ 
\begin{equation}\label{hardy-sobolev}
    |\sprod{\ph}{V\ph}| \leq C\|V_1\|_{5/3} \|\ph\|_{H^1}^2 + \|V_2\|_{\infty} \|\ph\|^2. 
\end{equation}
In particular, $V$ is infinitesimally form bounded w.r.to $-\Delta$.
\end{lemma}

\begin{proof}
(i) Let $\Omega\subset\R^3$ be compact. By Cauchy-Schwarz,
\begin{align*}
\int_{\Omega} |A_{\alpha^{-1}}(x)|^2dx &= \alpha\int_{\alpha^{-1}\Omega} |A(x)|^2dx\\
&\le \left(\int |A(x)|^3\chi_{\alpha^{-1}\Omega} (x) dx\right)^{2/3} |\Omega|^{1/3}\to 0\quad (\alpha\to\infty).
\end{align*}
The second statement of (i) is proved similarly.

In statement (ii) the contribution due to $V_2$ is obvious. Let us assume that
$V=V_1\in L^{5/3}(\R^3)$. By H\"older`s inequality $ |\sprod{\ph}{V\ph}|
\leq \|V\|_{5/3} \|\ph\|_5^2$ and
\begin{equation}
     \int |\ph|^5\,dx \leq \|\ph\|^{1/2}\left(\int |\ph|^6\,dx\right)^{3/4}.
\end{equation}
Using the general inequality $ab\leq p^{-1}a^p+q^{-1}b^q$ with $p=10$ and $q=10/9$ we obtain
$$
    \|\ph\|_5^2 \leq \|\ph\|^{1/5} \|\ph\|_6^{9/5}\leq
    \frac{1}{10}\|\ph\|^2 + \frac{9}{10} \|\ph\|_6^2.
$$
Statement (ii) now follows from the Sobolev inequality $ \|\ph\|_6^2
\leq C\|\nabla\ph\|^2$.  The infinitesimal form bound follows from the
fact that the norm of the $L^{5/3}$-part of $V$ can be chosen arbitrarily small.
\end{proof}

\begin{lemma}\label{lm:weak-field2}
Let $A$, $V$ be real-valued potentials satisfying the hypothesis of
Lemma~\ref{lm:weakfieldlimit} (i), and suppose that \eqref{V-Delta-bound} holds. 
If $\lim_{\alpha\to\infty} \lambda (\alpha) = 1$, then
\begin{itemize}
\item[(a)] $\lim_{\alpha\to\infty} \EE_{\lambda^2} (A_{\alpha^{-1}},\lambda V_{\alpha^{-1}}, \ph) = \EE (\ph)$ for all $\ph\in C_0^{\infty} (\R^3)$,
\item[(b)] $\lim_{\alpha\to\infty} E_P(A_{\alpha^{-1}},\lambda V_{\alpha^{-1}}, \lambda ^2) = E_P$.
\end{itemize}
If $A=(B\wedge x)/2$, $V\in L^{5/3}(\R^3) +L^{\infty}(\R^3)$ and $\lim_{\alpha\to\infty} \lambda (\alpha) = 1$, then 
\begin{itemize}
\item[(c)] $E_P(A_{\alpha^{-1}},\lambda V_{\alpha^{-1}}, \lambda ^2)
  =\lambda^4 E_P + O(\alpha^{-1/5}),\quad (\alpha\to\infty)$.
\end{itemize}
\end{lemma}

\begin{proof}
(a) For $\ph\in C_0^{\infty}(\R^3)$, Lemma~\ref{lm:weakfieldlimit} implies that 
$\|A_{\alpha^{-1}}\ph\|\to 0$ and $\sprod{\ph}{V_{\alpha^{-1}}\ph } \to 0$ as $\alpha\to\infty$.
This proves (a).

(b) For any normalized $\ph\in C_0^{\infty}(\R^3)$, by (a),
\begin{align*}
   \limsup _{\alpha\to\infty}  E_P(A_{\alpha^{-1}},\lambda
   V_{\alpha^{-1}},\lambda^2) \leq  \limsup _{\alpha\to\infty}
   \EE_{\lambda^2} (A_{\alpha^{-1}},\lambda V_{\alpha^{-1}},\ph) = \EE(\ph).
\end{align*}
This implies that $\limsup _{\alpha\to\infty}  E_P(A,\lambda V, \lambda^2 \alpha) \alpha^{-2}\le E_P$.

For (b) it remains to prove that $\liminf_{\alpha\to\infty}
E_P(A_{\alpha^{-1}},\lambda V_{\alpha^{-1}},\lambda^2) \ge E_P$. 
By the hypothesis on $V$, for any $\eps>0$ and any normalized $\ph\in
C_0^{\infty}(\R^3)$,
\begin{equation}\label{vrelbesch}
    |\sprod{\ph}{V_{\alpha^{-1}}\ph}| \leq \eps\|\nabla|\ph|\|^2+ \frac{C_{\eps}}{ \alpha^2}.
\end{equation}
From \eqref{vrelbesch}, the diamagnetic inequality and
the scaling property of $E_P$, it follows that 
\begin{align}
\EE_{\lambda^2}(A_{\alpha^{-1}}, \lambda V_{\alpha^{-1}}, \ph ) &\ge
(1-\lambda\eps) E_P(0,0,\lambda^2/(1-\lambda\eps)) - \lambda\frac{C_{\eps}}{\alpha^2} \label{ineq:energy1}\\
&\geq \frac{\lambda^4}{(1-\lambda\eps)} E_P - \lambda\frac{C_\eps}{\alpha^2},
\end{align}
and hence that 
$$
      E_P(A_{\alpha^{-1}},\lambda V_{\alpha^{-1}},\lambda^2) \geq \frac{\lambda^4}{(1-\lambda\eps)} E_P - \lambda\frac{C_\eps}{\alpha^2}.
$$
Now letting first $\alpha\to\infty$ and then $\eps\to 0$, the desired lower bound is obtained.

The proof of the lower bound in (c) is similar to the proof of the
lower bound in (b), the main difference being that we now have \eqref{hardy-sobolev}
from Lemma~\ref{lm:weakfieldlimit}, which implies that 
\begin{equation}\label{V-1-bound}
     |\sprod{\ph}{V_{\alpha^{-1}}\ph}| \leq
     C\alpha^{-1/5}\big(\|\ph\|^2 + \|\nabla|\ph|\|^2 \big)
\end{equation}
with some $C>0$ that is independent of $\alpha$ and $\ph$. By the
diamagnetic inequality and by \eqref{V-1-bound}, for any normalized
$\ph\in C_0^{\infty}(\R^3)$,
\begin{align*}
      \EE_{\lambda^2}(A_{\alpha^{-1}}, \lambda V_{\alpha^{-1}}, \ph)
      &\geq \EE_{\lambda^2}(0,0,|\ph|) - C\lambda \alpha^{-1/5}\big(1+\|\nabla|\ph|\|^2\big)\\
     &\geq  (1- C\lambda \alpha^{-1/5}) E_P(0,0,\frac{\lambda^2}{1-C\lambda \alpha^{-1/5}}) - C\lambda \alpha^{-1/5}\\
    &=\frac{\lambda^4}{1-C\lambda \alpha^{-1/5}} E_P -  C\lambda \alpha^{-1/5}.
\end{align*}
Hence $E_P(A_{\alpha^{-1}},\lambda V_{\alpha^{-1}},\lambda^2) \geq
\lambda^4 E_P- O(\alpha^{-1/5})$.

It remains to prove the upper bound on  $E_P(A_{\alpha^{-1}},\lambda
V_{\alpha^{-1}}, \lambda ^2) $ in (c). To this end let $\ph_0$ be a
(real-valued) minimizer of the Pekar functional \cite{Lieb1977},
i.e. $\EE(\ph_0) = E_P$ and let $\ph_\lambda$ be scaled in such a way
that  $\EE_{\lambda^2}(\ph_\lambda) = \lambda^4 \EE(\ph_0)$. Then 
\begin{align*}
     E_P(A_{\alpha^{-1}},\lambda V_{\alpha^{-1}}, \lambda ^2) & \leq
     \EE_{\lambda^2} (A_{\alpha^{-1}},\lambda
     V_{\alpha^{-1}},\ph_\lambda) \\
    & =\lambda^4 E_P +  \|A_{\alpha^{-1}}\ph_\lambda\|^2 +\lambda
    \sprod{\ph_\lambda}{V_{\alpha^{-1}}\ph_\lambda} \\
    & = \lambda^4 E_P + O(\alpha^{-1/5}).
\end{align*}
We have used that
$\Rea\sprod{-i\nabla\ph_\lambda}{A_{\alpha^{-1}}\ph_\lambda}=0$, since
$\ph_\lambda$ is real-valued, and \eqref{hardy-sobolev} from Lemma~\ref{lm:weakfieldlimit}
\end{proof}

\section{Existence of Bipolarons}

Let $A,V$ be vector and scalar potentials, respectively, satisfying the assumptions of the Section~\ref{sec:strong-field}. Let $\alpha, U>0$. We define a two-body Hamiltonian 
$H_U^{A,V}$ on $L^2(\R^6)$ by 
\begin{equation}\label{2-body}
      H_U^{A,V} := (D_A^2+V)\otimes 1 + 1\otimes (D_A^2+V) + UV_C
\end{equation}
where $V_C(x,y):=|x-y|^{-1}$. More precisely, we define $H_U^{A,V}$ in terms of the quadratic form given by the right hand side of 
\eqref{2-body} on $C_0^{\infty}(\R^6)$. Its form domain will be denoted by $H^1_{A}(\R^6)$.  

In the two-polaron model of Fr\"ohlich, the minimal energy, $E_{2}(A,V,U,\alpha)$ of two electrons in a polar crystal is the infimum of the quadratic form
\begin{equation}\label{2-el-F}
     \sprod{\psi}{(H_U^{A,V}\otimes 1) \psi} +N(\psi) +\sqrt{\alpha}W_2(\psi),
\end{equation}
whose domain is the intersection of $H^1_{A}(\R^6)\otimes \FF_0$ with the unit sphere of the Hilbert space $L^2(\R^6)\otimes \FF$. Here $N(\psi)$ and $W_2(\psi)$ are defined by expressions similar to \eqref{number} and \eqref{e-p-inter}, the main difference being that $e^{ikx}$ in \eqref{e-p-inter} becomes $e^{ikx_1}+e^{ikx_2}$ in $W_2(\psi)$. 

In the two-polaron model of Pekar and Tomasevich  the minimal energy, $E_{PT}(A,V, U,\alpha)$ of two electrons in a polar crystal is the infimum of the functional
$$
    \sprod{\ph}{H_U^{A,V} \ph}  -\alpha \int\!\!\int\frac{\rho(x)\rho(y)}{|x-y|}\, dxdy
$$
on the $L^2$-unit sphere of $H^1_{A}(\R^6)$, where
$\rho(x):= \int \big(|\ph(x,y)|^2+ |\ph(y,x)|^2\big)\, dy$. For any
fixed $\ph\in L^2(\R^6)\cap H^1_{A}(\R^6)$, $\|\ph\| =1$, and corresponding density $\rho$ the identity
$$
\inf_{\|\eta\| =1} \left(N(\ph\otimes\eta) +\sqrt{\alpha} W_2(\ph\otimes \eta)\right) = - \alpha \int\!\!\int\frac{\rho(x)\rho(y)}{|x-y|}\, dxdy
$$
holds. By choosing $\psi = \ph\otimes\eta$ in \eqref{2-el-F} it follows that
\begin{equation}
    E_{PT}(A,V,U,\alpha) \geq E_{2}(A,V,U,\alpha), \label{inequ:pt-froe}
\end{equation}
which, together with Theorem~\ref{limit-thm} and the results of
\cite{GHW2012} enables us to prove the following theorem on the
\emph{binding of polarons}:

\begin{theorem}\label{binding-thm}
Suppose $A,V\in L^2_{\rm loc}(\R^3)$ and that $V$ is infinitesimally
operator bounded w.r.to $\Delta$. Let $A_\alpha(x)=\alpha A(\alpha x)$
and $V_\alpha(x) = \alpha^2 V(\alpha x)$. If the Pekar functional
\eqref{Pekar-AV} attains its minimum, then there exists $u_{A,V}>2$
such that for $U<\alpha u_{A,V}$ and $\alpha$ large enough
$$
       2 E(A_\alpha,V_\alpha,\alpha) > E_{2}(A_\alpha,V_\alpha,U,\alpha).
$$
\end{theorem}

\begin{proof}
Let $U=\alpha u$. By a simple scaling argument
\begin{equation}\label{pekartom-scaling}
E_{PT} (A_{\alpha}, V_{\alpha}, \alpha u,\alpha) = \alpha^2 E_{PT} (A,V,u,1), 
\end{equation}
which is analogous to \eqref{scaling-pekar}. By Theorem~3.1 of \cite{GHW2012} there exists $u_{A,V}>2$ such that for $u<u_{A,V}$,
\begin{equation}
2E_P(A,V) >E_{PT}(A,V,u,1) .\label{pekar-binding}
\end{equation}
From Theorem~\ref{limit-thm}, \eqref{pekar-binding}, \eqref{pekartom-scaling}, and \eqref{inequ:pt-froe} it follows that, for $\alpha$ large enough,
\begin{align*}
2\alpha^{-2} E(A_{\alpha},V_{\alpha},\alpha) &= 2 E_P(A,V) - o(1)\\
&> E_{PT} (A,V,u,1)\\
&= \alpha^{-2} E_{PT} (A_{\alpha},V_{\alpha},\alpha u,\alpha)\\
&\ge  \alpha^{-2} E_{2} (A_{\alpha},V_{\alpha},U,\alpha),
\end{align*}
which proves the theorem.
\end{proof}


\begin{thebibliography}{10}

\bibitem{AL2012}
Ioannis Anapolitanos and Landon Benjamin.
\newblock The ground state energy of the multi-polaron in the strong coupling
  limit.
\newblock arXiv:1212.3571 [math-ph], 2012.

\bibitem{AG2013}
Ioannis Anapolitanos and Marcel Griesemer.
\newblock Multipolarons in a constant magnetic field.
\newblock arXiv:1204.5660, [math-ph], 2013.

\bibitem{BD1996}
F.~Brosens and J.~T. Devreese.
\newblock Stability of bipolarons in the presence of a magnetic field.
\newblock {\em Phys. Rev. B}, 54(14):9792--9808, Oct 1996.

\bibitem{DV1983}
M.~D. Donsker and S.~R.~S. Varadhan.
\newblock Asymptotics for the polaron.
\newblock {\em Comm. Pure Appl. Math.}, 36(4):505--528, 1983.

\bibitem{FG2013}
R.~L. Frank and L.~Geisinger.
\newblock Ground state energy of a polaron in a strong magnetic field.
\newblock arXiv:1303.2382, [math-ph], 2013.

\bibitem{FLST2012}
Rupert~L. Frank, Elliott~H. Lieb, Robert Seiringer, and Lawrence~E. Thomas.
\newblock Ground state properties of multi-polaron systems.
\newblock arXiv:1209.3717, [math-ph], 2012.

\bibitem{GHW2012}
M.~Griesemer, F.~Hantsch, and D.~Wellig.
\newblock On the magnetic pekar functional and the existence of bipolarons.
\newblock {\em Reviews in Mathematical Physics}, 24(06):1250014, 2012.

\bibitem{Lew2011}
Mathieu Lewin.
\newblock Geometric methods for nonlinear many-body quantum systems.
\newblock {\em J. Funct. Anal.}, 260(12):3535--3595, 2011.

\bibitem{Lieb1977}
Elliott~H. Lieb.
\newblock Existence and uniqueness of the minimizing solution of {C}hoquard's
  nonlinear equation.
\newblock {\em Studies in Appl. Math.}, 57(2):93--105, 1976/77.

\bibitem{LiebLoss}
Elliott~H. Lieb and Michael Loss.
\newblock {\em Analysis}, volume~14 of {\em Graduate Studies in Mathematics}.
\newblock American Mathematical Society, Providence, RI, second edition, 2001.

\bibitem{LT1997}
Elliott~H. Lieb and Lawrence~E. Thomas.
\newblock Exact ground state energy of the strong-coupling polaron.
\newblock {\em Comm. Math. Phys.}, 183(3):511--519, 1997.

\bibitem{Lions1984}
P.-L. Lions.
\newblock The concentration-compactness principle in the calculus of
  variations. {T}he locally compact case. {I}.
\newblock {\em Ann. Inst. H. Poincar\'e Anal. Non Lin\'eaire}, 1(2):109--145,
  1984.

\bibitem{MS2007}
Tadahiro Miyao and Herbert Spohn.
\newblock The bipolaron in the strong coupling limit.
\newblock {\em Ann. Henri Poincar\'e}, 8(7):1333--1370, 2007.

\bibitem{RS2}
Michael Reed and Barry Simon.
\newblock {\em Methods of modern mathematical physics. {II}. {F}ourier
  analysis, self-adjointness}.
\newblock Academic Press [Harcourt Brace Jovanovich Publishers], New York,
  1975.

\bibitem{VSPD1992}
G.~Verbist, M.~A. Smondyrev, F.~M. Peeters, and J.~T. Devreese.
\newblock Strong-coupling analysis of large bipolarons in two and three
  dimensions.
\newblock {\em Phys. Rev. B}, 45:5262--5269, Mar 1992.

\end{thebibliography}

\end{document}